\RequirePackage{amsmath}
\RequirePackage{fix-cm}
\documentclass[smallextended]{svjour3}
\smartqed
\usepackage{amsfonts}
\usepackage{cite}
\usepackage{graphicx}
\usepackage[colorlinks=true]{hyperref}

\journalname{Letters in Mathematical Physics}
\begin{document}

\title{Uncertainty relations with quantum memory for the Wehrl entropy}

\author{Giacomo De Palma}

\institute{Giacomo De Palma \at QMATH, Department of Mathematical Sciences, University of Copenhagen, Universitetsparken 5, 2100 Copenhagen, Denmark\\ Tel.: +45 35 33 68 04\\ \email{giacomo.depalma@math.ku.dk}}

\date{}
\maketitle

\begin{abstract}
We prove two new fundamental uncertainty relations with quantum memory for the Wehrl entropy.
The first relation applies to the bipartite memory scenario.
It determines the minimum conditional Wehrl entropy among all the quantum states with a given conditional von Neumann entropy, and proves that this minimum is asymptotically achieved by a suitable sequence of quantum Gaussian states.
The second relation applies to the tripartite memory scenario.
It determines the minimum of the sum of the Wehrl entropy of a quantum state conditioned on the first memory quantum system with the Wehrl entropy of the same state conditioned on the second memory quantum system, and proves that also this minimum is asymptotically achieved by a suitable sequence of quantum Gaussian states.
The Wehrl entropy of a quantum state is the Shannon differential entropy of the outcome of a heterodyne measurement performed on the state.
The heterodyne measurement is one of the main measurements in quantum optics, and lies at the basis of one of the most promising protocols for quantum key distribution.
These fundamental entropic uncertainty relations will be a valuable tool in quantum information, and will e.g. find application in security proofs of quantum key distribution protocols in the asymptotic regime and in entanglement witnessing in quantum optics.
\end{abstract}

\section{Introduction}
Entropic uncertainty relations provide a lower bound to the sum of the entropies of the outcomes of two incompatible measurements performed on the same quantum state.
Entropic uncertainty relations are a fundamental tool of quantum information theory, since they are the central ingredient in the security analysis of almost all quantum cryptographic protocols and can be used for entanglement witnessing (see \cite{coles2017entropic} for a review).
The quantum information community has recently focused on the scenario with quantum memory \cite{berta2010uncertainty,coles2012uncertainty,furrer2014position,coles2017entropic}, where the entropies are conditioned on the knowledge of external observers holding memory quantum systems correlated with the measured system.

The heterodyne measurement \cite{schleich2015quantum} is one of the main measurements in quantum optics.
It is used for quantum tomography \cite{carmichael2013statistical}, and lies at the basis of one of the most promising quantum key distribution protocols \cite{weedbrook2012gaussian,weedbrook2004quantum}.
The Wehrl entropy \cite{wehrl1979relation,wehrl1978general} of a quantum state is the Shannon differential entropy \cite{cover2006elements} of the outcome of a heterodyne measurement performed on the state.
The elements of the POVM that models the heterodyne measurement are the projectors onto the coherent states \cite{schrodinger1926stetige,bargmann1961hilbert,klauder1960action,glauber1963coherent,klauder2006fundamentals}.
Since the coherent states are not orthogonal, the associated projectors do not commute, and nontrivial entropic uncertainty relations are allowed for the heterodyne measurement alone.

The basic uncertainty relation for the Wehrl entropy states that its minimum is achieved by the vacuum state \cite{lieb1978proof,carlen1991some}. This relation has been recently improved: it has been proven that thermal quantum Gaussian states minimize the Wehrl entropy among all the quantum states with a given von Neumann entropy \cite{de2017wehrl,de2017multimode}.
In this paper, we prove two new fundamental uncertainty relations for the Wehrl entropy in the scenario with quantum memory.
The first relation (\autoref{thm:ineq}) applies to the bipartite scenario with one memory quantum system.
It determines the minimum conditional Wehrl entropy among all the quantum states with a given conditional von Neumann entropy, and proves that this minimum is asymptotically achieved by a suitable sequence of quantum Gaussian states (\autoref{thm:optimbi}).
This sequence is built from a two-mode infinitely squeezed pure state shared between the system to be measured and the memory, applying the quantum-limited amplifier to the system to be measured.
The second relation (\autoref{thm:main}) applies to the tripartite scenario with two memory quantum systems.
It determines the minimum of the sum of the Wehrl entropy of a quantum state conditioned on the first memory quantum system with the Wehrl entropy of the same quantum state conditioned on the second memory quantum system, and proves that also this minimum is asymptotically achieved by a suitable sequence of quantum Gaussian states (\autoref{thm:optimtri}).
This sequence is the purification of the sequence that saturates the bipartite memory uncertainty relation, and the purifying system plays the role of the second memory.

The key ingredient of the proof of \autoref{thm:ineq} is the Entropy Power Inequality with quantum memory \cite{de2018conditional,koenig2015conditional}.
This fundamental inequality determines the minimum conditional von Neumann entropy of the output of the beam-splitter or of the squeezing among all the input states where the two inputs are conditionally independent given the memory and have given conditional von Neumann entropies.
This inequality generalizes the quantum Entropy Power Inequality \cite{de2014generalization,de2015multimode,konig2014entropy,konig2016corrections} to the scenario with quantum memory.
The link with the Wehrl entropy is provided by \autoref{thm:QA}, stating that the heterodyne measurement is asymptotically equivalent to the quantum-limited amplifier in the limit of infinite amplification parameter.
The proof of \autoref{thm:QA} is based on a new Berezin-Lieb inequality (\autoref{thm:LB}) for the scenario with quantum memory.

The fundamental uncertainty relations proven in this paper will be a valuable tool in quantum information and quantum cryptography.
Indeed, one of the most promising protocols for quantum key distribution is based on the exchange of Gaussian coherent states and on the heterodyne measurement \cite{weedbrook2004quantum}.
The security of a variant of this protocol has recently been proven \cite{leverrier2016p,leverrier2017security}.
This variant requires integrating in the protocol a symmetrisation procedure that is difficult to implement.
The security of the original protocol without the symmetrisation has not been proven yet.
With the uncertainty relations proven in this paper, it will be possible to prove the security of the original protocol in the asymptotic regime of a key of infinite length \cite{coles2017entropic}.
The proof might exploit the techniques of \cite{furrer2012continuous}, which proves the security of a quantum key distribution protocol based on the homodyne measurement through an entropic uncertainty relation for the joint measurement of position and momentum.
Among the other possible applications of our results, we mention e.g. entanglement witnessing (see \autoref{cor:ent}).

The paper is structured as follows.
In \autoref{sec:GQS}, we introduce Gaussian quantum systems and the heterodyne measurement.
In \autoref{sec:QA}, we prove the equivalence between heterodyne measurement and quantum-limited amplifier.
In \autoref{sec:bi} and \autoref{sec:tri}, we prove the bipartite and tripartite memory entropic uncertainty relations, respectively, and in \autoref{sec:opt} we prove their optimality.

\section{Gaussian quantum systems}\label{sec:GQS}
We consider the Hilbert space of $M$ harmonic oscillators, or $M$ modes of the electromagnetic radiation, i.e. the irreducible representation of the canonical commutation relations
\begin{equation}\label{eq:CCa}
\left[\hat{a}_i,\;\hat{a}_j^\dag\right]=\delta_{ij}\;\hat{\mathbb{I}}\;,\qquad i,\,j=1,\ldots,\,M\;.
\end{equation}
The operators $\hat{a}_1^\dag\hat{a}_1,\ldots,\,\hat{a}_M^\dag\hat{a}_M$ have integer spectrum and commute.
Their joint eigenbasis is the Fock basis $\left\{|n_1\ldots n_M\rangle\right\}_{n_1,\ldots,\,n_M\in\mathbb{N}}$.
The Hamiltonian
\begin{equation}
\hat{N}=\sum_{i=1}^M\hat{a}_i^\dag\hat{a}_i
\end{equation}
counts the number of excitations, or photons.
We define the quadratures
\begin{equation}
\hat{Q}_i = \frac{\hat{a}_i + \hat{a}_i^\dag}{\sqrt{2}}\;,\qquad\hat{P}_i = \frac{\hat{a}_i - \hat{a}_i^\dag}{\mathrm{i}\sqrt{2}}\;,\qquad i=1,\,\ldots,\,M\;.
\end{equation}
We can collect them in the vector
\begin{equation}
\hat{R}_{2i-1} = \hat{Q}_i\;,\qquad\hat{R}_{2i} = \hat{P}_i\;,\qquad i=1,\,\ldots,\,M\;,
\end{equation}
and \eqref{eq:CCa} becomes
\begin{equation}
\left[\hat{R}_i,\;\hat{R}_{j}\right] = \mathrm{i}\,\Delta_{ij}\,\hat{\mathbb{I}}\;,\qquad i,\,j=1,\,\ldots,\,2M\;,
\end{equation}
where
\begin{equation}
\Delta = \bigoplus_{i=1}^M\left(
                             \begin{array}{cc}
                               0 & 1 \\
                               -1 & 0 \\
                             \end{array}
                           \right)
\end{equation}
is the symplectic form.

\subsection{Quantum Gaussian states}
A quantum Gaussian state is a density operator proportional to the exponential of a quadratic polynomial in the quadratures:
\begin{equation}
\hat{\gamma} = \frac{\exp\left(-\frac{1}{2}\sum_{i,\,j=1}^{2M}h_{ij}\,\hat{R}_i\,\hat{R}_j\right)} {\mathrm{Tr}\exp\left(-\frac{1}{2}\sum_{i,\,j=1}^{2M}h_{ij}\,\hat{R}_i\,\hat{R}_j\right)}\;,
\end{equation}
where $h$ is a positive real $2M\times2M$ matrix.
A quantum Gaussian state is completely identified by its covariance matrix
\begin{equation}
\sigma_{ij} = \frac{1}{2}\,\mathrm{Tr}\left[\left(\hat{R}_i\,\hat{R}_j + \hat{R}_j\,\hat{R}_i\right)\hat{\gamma}\right]\;,\qquad i,\,j=1,\,\ldots,\,2M\;.
\end{equation}
The von Neumann entropy of a quantum Gaussian state is
\begin{equation}
S = \sum_{i=1}^M g\left(\nu_i-\frac{1}{2}\right)\;,
\end{equation}
where
\begin{equation}
g(x)=\left(x+1\right)\ln\left(x+1\right) - x\ln x\;,
\end{equation}
and $\nu_1,\,\ldots,\,\nu_M$ are the symplectic eigenvalues of its covariance matrix $\sigma$, i.e., the absolute values of the eigenvalues of $\sigma\Delta^{-1}$.

\subsection{Coherent states}
The classical phase space associated with a $M$-mode Gaussian quantum system is $\mathbb{C}^M$, and for any $\mathbf{z}\in\mathbb{C}^M$ we define the coherent state
\begin{equation}
|\mathbf{z}\rangle = \mathrm{e}^{-\frac{|\mathbf{z}|^2}{2}}\sum_{n_1,\,\ldots,\,n_M=0}^\infty\frac{z_1^{n_1}\ldots z_M^{n_M}}{\sqrt{n_1!\ldots n_M!}}\;|n_1\ldots n_M\rangle\;.
\end{equation}
Coherent states satisfy the resolution of the identity \cite{holevo2015gaussian}
\begin{equation}\label{eq:complz}
\int_{\mathbb{C}^M}|\mathbf{z}\rangle\langle \mathbf{z}|\,\frac{\mathrm{d}^{2M}z}{\pi^M} = \hat{\mathbb{I}}\;,
\end{equation}
where the integral converges in the weak topology.
The POVM associated with the resolution of the identity \eqref{eq:complz} is called heterodyne measurement \cite{schleich2015quantum}.

\subsection{The Gaussian quantum-limited amplifier}\label{sec:A}
The $M$-mode Gaussian quantum-limited amplifier $\mathcal{A}_{\kappa}$ with amplification parameter $\kappa\ge1$ performs a two-mode squeezing on the input state $\hat{\rho}$ and the vacuum state of an $M$-mode ancillary Gaussian system $E$ with ladder operators $\hat{e}_1,\ldots,\,\hat{e}_M$:
\begin{equation}\label{eq:A}
\mathcal{A}_{\kappa}(\hat{\rho})= \mathrm{Tr}_E\left[\hat{U}_\kappa\left(\hat{\rho}\otimes|\mathbf{0}\rangle\langle\mathbf{0}|\right)\hat{U}_\kappa^\dag\right]\;.
\end{equation}
The squeezing unitary operator
\begin{equation}\label{eq:defUk}
\hat{U}_\kappa=\exp\left(\mathrm{arccosh}\sqrt{\kappa}\sum_{i=1}^M\left(\hat{a}_i^\dag\hat{e}_i^\dag-\hat{a}_i\,\hat{e}_i\right)\right)
\end{equation}
acts on the ladder operators as
\begin{subequations}
\begin{align}
\hat{U}_\kappa^\dag\;\hat{a}_i\;\hat{U}_\kappa &= \sqrt{\kappa}\;\hat{a}_i + \sqrt{\kappa-1}\;\hat{e}_i^\dag\;,\\
\hat{U}_\kappa^\dag\;\hat{e}_i\;\hat{U}_\kappa &= \sqrt{\kappa-1}\;\hat{a}_i^\dag + \sqrt{\kappa}\;\hat{e}_i\;,\qquad i=1,\ldots,\,M\;.
\end{align}
\end{subequations}
The Gaussian quantum-limited amplifier acts on quantum Gaussian states as follows.
Let $A$ and $B$ be Gaussian quantum systems with $M_A$ and $M_B$ modes, respectively.
Let $\hat{\gamma}_{AB}$ be the joint quantum Gaussian state on $AB$ with covariance matrix
\begin{equation}
\sigma_{AB} = \frac{1}{2}\left(
                 \begin{array}{cc}
                   X & Z \\
                   Z^T & Y \\
                 \end{array}
               \right)\;.
\end{equation}
Then, $(\mathcal{A}_\kappa\otimes\mathbb{I}_B)(\hat{\gamma}_{AB})$ is the quantum Gaussian state with covariance matrix
\begin{equation}
\sigma_{AB}' = \frac{1}{2}\left(
                 \begin{array}{cc}
                   \kappa\,X + \left(\kappa-1\right)I_{M_A} & \sqrt{\kappa}\,Z \\
                   \sqrt{\kappa}\,Z^T & Y \\
                 \end{array}
               \right)\;.
\end{equation}

\section{Asymptotic equivalence between heterodyne measurement and amplifier}\label{sec:QA}
In this Section, we extend the asymptotic equivalence between the heterodyne measurement and the quantum-limited amplifier proven in \cite{de2017wehrl} to the scenario with quantum memory.
\begin{theorem}[heterodyne measurement - amplifier equivalence]\label{thm:QA}
Let $A$ be an $M$-mode Gaussian quantum system and $B$ a generic quantum system.
For any concave function $f:[0,1]\to[0,\infty)$ and for any joint quantum state $\hat{\rho}_{AB}$ on $AB$,
\begin{equation}
\int_{\mathbb{C}^M}\mathrm{Tr}_B f\left(\langle\mathbf{z}|\hat{\rho}_{AB}|\mathbf{z}\rangle\right) \ge \limsup_{\kappa\to\infty} \frac{\mathrm{Tr}_{AB} f\left(\kappa^M(\mathcal{A}_\kappa\otimes\mathbb{I}_B)(\hat{\rho}_{AB})\right)}{\kappa^M}\;.
\end{equation}
\end{theorem}

\subsection{Proof of \autoref{thm:QA}}
The proof is based on the following generalization of the Berezin-Lieb inequality \cite{berezin1972covariant} to the scenario with quantum memory.
\begin{theorem}[Berezin-Lieb inequality with quantum memory]\label{thm:LB}
Let $A$ be an $M$-mode Gaussian quantum system and $B$ a generic quantum system.
Then, for any trace-class operator $\hat{X}$ on $AB$ with $0\le\hat{X}\le\hat{\mathbb{I}}_{AB}$ and any concave function $f:[0,1]\to[0,\infty)$,
\begin{equation}
\int_{\mathbb{C}^M}\mathrm{Tr}_B f\left(\langle\mathbf{z}|\hat{X}|\mathbf{z}\rangle\right)\frac{\mathrm{d}^{2M}z}{\pi^M} \ge \mathrm{Tr}_{AB}f\left(\hat{X}\right)\;.
\end{equation}
\end{theorem}
\begin{proof}
Let us diagonalize $\hat{X}$:
\begin{equation}
\hat{X} = \sum_{k=0}^\infty x_k\,|\psi_k\rangle\langle\psi_k|\;, \quad\langle\psi_k|\psi_l\rangle=\delta_{kl}\quad\forall\;k,\,l\in\mathbb{N}\;, \quad\sum_{k=0}^\infty|\psi_k\rangle\langle\psi_k|=\hat{\mathbb{I}}_{AB}\;.
\end{equation}
We have
\begin{equation}
0\le x_k \le 1\qquad\forall\;k\in\mathbb{N}\;,\qquad \sum_{k=0}^\infty x_k<\infty\;.
\end{equation}
From the completeness of the $|\psi_k\rangle$ we have for any $\mathbf{z}\in\mathbb{C}^M$
\begin{equation}
\sum_{k=0}^\infty \langle\mathbf{z}|\psi_k\rangle\langle\psi_k|\mathbf{z}\rangle = \langle\mathbf{z}|\hat{\mathbb{I}}_{AB}|\mathbf{z}\rangle = \hat{\mathbb{I}}_B\;.
\end{equation}
We can then apply \autoref{lem:jensen} to $|\phi_k\rangle = \langle\mathbf{z}|\psi_k\rangle$ and get
\begin{align}
\mathrm{Tr}_B f\left(\langle\mathbf{z}|\hat{X}|\mathbf{z}\rangle\right) &= \mathrm{Tr}_B f\left(\sum_{k=0}^\infty x_k\,\langle\mathbf{z}|\psi_k\rangle \langle\psi_k|\mathbf{z}\rangle\right) \ge \sum_{k=0}^\infty\langle\psi_k|\mathbf{z}\rangle\langle\mathbf{z}|\psi_k\rangle f(x_k)\nonumber\\
&=\mathrm{Tr}_B\langle\mathbf{z}|f\left(\hat{X}\right)|\mathbf{z}\rangle\;.
\end{align}
Finally, from the completeness relation \eqref{eq:complz} we have
\begin{equation}
\int_{\mathbb{C}^M}\mathrm{Tr}_B f\left(\langle\mathbf{z}|\hat{X}|\mathbf{z}\rangle\right)\frac{\mathrm{d}^{2M}z}{\pi^M} \ge \int_{\mathbb{C}^M}\mathrm{Tr}_B\langle\mathbf{z}|f\left(\hat{X}\right)|\mathbf{z}\rangle\,\frac{\mathrm{d}^{2M}z}{\pi^M} = \mathrm{Tr}_{AB}f\left(\hat{X}\right)\;.
\end{equation}
\end{proof}
From \cite{ivan2011operator}, Theorem 9, the map $\kappa^M\mathcal{A}_\kappa$ is unital, and
\begin{equation}
0 \le \kappa^M(\mathcal{A}_\kappa\otimes\mathbb{I}_B)(\hat{\rho}_{AB}) \le \kappa^M(\mathcal{A}_\kappa\otimes\mathbb{I}_B)\left(\hat{\mathbb{I}}_{AB}\right) = \hat{\mathbb{I}}_{AB}\;.
\end{equation}
We can then apply \autoref{thm:LB} to $\hat{X}=\kappa^M(\mathcal{A}_\kappa\otimes\mathbb{I}_B)(\hat{\rho}_{AB})$ and get
\begin{equation}\label{eq:1}
\int_{\mathbb{C}^M}\mathrm{Tr}_B f\left(\kappa^M\langle\mathbf{z}|(\mathcal{A}_\kappa\otimes\mathbb{I}_B)(\hat{\rho}_{AB})|\mathbf{z}\rangle\right)\frac{\mathrm{d}^{2M}z}{\pi^M} \ge \mathrm{Tr}_{AB} f\left(\kappa^M(\mathcal{A}_\kappa\otimes\mathbb{I}_B)(\hat{\rho}_{AB})\right)\;.
\end{equation}
Since for any $\mathbf{z}\in\mathbb{C}^M$
\begin{equation}\label{eq:adj}
\kappa^M\langle\mathbf{z}|(\mathcal{A}_\kappa\otimes\mathbb{I}_B)(\hat{\rho}_{AB})|\mathbf{z}\rangle = \langle\mathbf{z}/\sqrt{\kappa}|\hat{\rho}_{AB}|\mathbf{z}/\sqrt{\kappa}\rangle
\end{equation}
(\cite{de2017wehrl}, Lemma 4), \eqref{eq:1} becomes
\begin{equation}
\int_{\mathbb{C}^M}\mathrm{Tr}_B f\left(\langle\mathbf{z}|\hat{\rho}_{AB}|\mathbf{z}\rangle\right)\frac{\mathrm{d}^{2M}z}{\pi^M} \ge \frac{1}{\kappa^M}\mathrm{Tr}_{AB} f\left(\kappa^M(\mathcal{A}_\kappa\otimes\mathbb{I}_B)(\hat{\rho}_{AB})\right)\;,
\end{equation}
and the claim follows taking the limit $\kappa\to\infty$.

\section{Bipartite quantum memory uncertainty relation for the Wehrl entropy}\label{sec:bi}
\begin{theorem}[bipartite quantum memory uncertainty relation for the Wehrl entropy]\label{thm:ineq}
Let $A$ be an $M$-mode Gaussian quantum system and $B$ a generic quantum system.
Let $\hat{\rho}_{AB}$ be a joint quantum state on $AB$ such that its marginal $\hat{\rho}_A=\mathrm{Tr}_B\hat{\rho}_{AB}$ has finite average energy, and its marginal $\hat{\rho}_B=\mathrm{Tr}_A\hat{\rho}_{AB}$ has finite von Neumann entropy.
Let $\hat{\rho}_{ZB}$ be the probability measure on $\mathbb{C}^M$ taking values in positive operators on $B$ associated with the heterodyne measurement on $A$:
\begin{equation}\label{eq:defrhoZ}
\mathrm{d}\hat{\rho}_{ZB}(\mathbf{z}) = \langle\mathbf{z}|\hat{\rho}_{AB}|\mathbf{z}\rangle\,\frac{\mathrm{d}^{2M}z}{\pi^M}\;,\qquad \mathbf{z}\in\mathbb{C}^M\;.
\end{equation}
Then, the following bipartite quantum memory uncertainty relation holds:
\begin{equation}\label{eq:biEUR}
S(Z|B)_{\hat{\rho}_{ZB}} \ge M\ln\left(\exp\frac{S(A|B)_{\hat{\rho}_{AB}}}{M}+1\right) \ge 0\;,
\end{equation}
where
\begin{equation}\label{eq:defZ|B}
S(Z|B)_{\hat{\rho}_{ZB}} = S(\hat{\rho}_{ZB}) - S(\hat{\rho}_B)
\end{equation}
is the Shannon differential entropy of the outcome $Z$ of the heterodyne measurement performed on $A$ conditioned on the quantum system $B$, and
\begin{equation}
S(\hat{\rho}_{ZB}) = -\int_{\mathbb{C}^M}\mathrm{Tr}_B\left[\langle\mathbf{z}|\hat{\rho}_{AB}|\mathbf{z}\rangle\ln\langle\mathbf{z}|\hat{\rho}_{AB}|\mathbf{z}\rangle\right] \frac{\mathrm{d}^{2M}z}{\pi^M}
\end{equation}
is the joint entropy of $Z$ and $B$.
\end{theorem}
\begin{remark}
If $B$ is the trivial system, \eqref{eq:biEUR} becomes
\begin{equation}\label{eq:monoEUR}
S(\hat{\rho}_Z) \ge M\ln\left(\exp\frac{S(\hat{\rho}_A)}{M}+1\right) \ge 0\;.
\end{equation}
However, the optimal inequality satisfied by the unconditioned Wehrl entropy is \cite{de2017wehrl}
\begin{equation}
S(\hat{\rho}_Z) \ge M\ln\left(g^{-1}\left(\frac{S(\hat{\rho}_A)}{M}\right)+1\right)+M \ge M\;,
\end{equation}
that is strictly stronger than \eqref{eq:monoEUR}.
The presence of the quantum memory then both changes the form of the optimal inequality and reduces the minimum uncertainty: while the minimum Wehrl entropy is $M$, the minimum conditional Wehrl entropy is $0$.
\end{remark}
\begin{corollary}[entanglement witnessing]\label{cor:ent}
Under the hypotheses of \autoref{thm:ineq}, if the quantum state $\hat{\rho}_{AB}$ is separable, the following bipartite quantum memory uncertainty relation holds:
\begin{equation}
S(Z|B)_{\hat{\rho}_{ZB}} \ge M\ln 2\;.
\end{equation}
\end{corollary}
\begin{proof}
Since $\hat{\rho}_{AB}$ is separable, we have $S(A|B)_{\hat{\rho}_{AB}}\ge0$.
The claim then follows from \autoref{thm:ineq}.
\end{proof}

\subsection{Proof of \autoref{thm:ineq}}
Applying \autoref{thm:QA} to $f(x)=-x\ln x$ we get
\begin{equation}
S(\hat{\rho}_{ZB}) \ge \limsup_{\kappa\to\infty}\left(S((\mathcal{A}_\kappa\otimes\mathbb{I}_B)(\hat{\rho}_{AB})) - M\ln\kappa\right)\;.
\end{equation}
Subtracting $S(\hat{\rho}_B)$ on both sides we get
\begin{equation}
S(Z|B)_{\hat{\rho}_{ZB}} \ge \limsup_{\kappa\to\infty}\left(S(A|B)_{(\mathcal{A}_\kappa\otimes\mathbb{I}_B)(\hat{\rho}_{AB})} - M\ln\kappa\right)\;.
\end{equation}
From the Entropy Power Inequality with quantum memory \cite{de2018conditional},
\begin{equation}\label{eq:EPI}
S(A|B)_{(\mathcal{A}_\kappa\otimes\mathbb{I}_B)(\hat{\rho}_{AB})} \ge M\,\ln\left(\kappa\exp\frac{S(A|B)_{\hat{\rho}_{AB}}}{M} + \kappa-1\right)\;.
\end{equation}
We then have
\begin{align}
S(Z|B)_{\hat{\rho}_{ZB}} &\ge M\limsup_{\kappa\to\infty}\ln\left(\exp\frac{S(A|B)_{\hat{\rho}_{AB}}}{M}+1-\frac{1}{\kappa}\right)\nonumber\\
&= M\ln\left(\exp\frac{S(A|B)_{\hat{\rho}_{AB}}}{M}+1\right)\;.
\end{align}

\section{Tripartite quantum memory uncertainty relation for the Wehrl entropy}\label{sec:tri}
\begin{theorem}[tripartite quantum memory uncertainty relation for the Wehrl entropy]\label{thm:main}
Let $A$ be an $M$-mode Gaussian quantum system, and let $B$, $C$ be arbitrary quantum systems.
Let $\hat{\rho}_{ABC}$ be a joint quantum state on $ABC$ such that its marginal $\hat{\rho}_A=\mathrm{Tr}_{BC}\hat{\rho}_{ABC}$ has finite average energy, and its marginals $\hat{\rho}_B=\mathrm{Tr}_{AC}\hat{\rho}_{ABC}$ and $\hat{\rho}_C=\mathrm{Tr}_{AB}\hat{\rho}_{ABC}$ have finite von Neumann entropy.
Let $\hat{\rho}_{ZBC}$ be the probability measure on $\mathbb{C}^M$ taking values in positive operators on $BC$ associated with the heterodyne measurement on $A$:
\begin{equation}
\mathrm{d}\hat{\rho}_{ZBC}(\mathbf{z}) = \langle\mathbf{z}|\hat{\rho}_{ABC}|\mathbf{z}\rangle\,\frac{\mathrm{d}^{2M}z}{\pi^M}\;,\qquad \mathbf{z}\in\mathbb{C}^M\;.
\end{equation}
Then, the following tripartite quantum memory uncertainty relation holds:
\begin{equation}\label{eq:triEUR}
S(Z|B)_{\hat{\rho}_{ZB}} + S(Z|C)_{\hat{\rho}_{ZC}} \ge M\ln 4\;,
\end{equation}
where $S(Z|B)_{\hat{\rho}_{ZB}}$ and $S(Z|C)_{\hat{\rho}_{ZC}}$ are defined as in \eqref{eq:defZ|B}.
\end{theorem}
\begin{proof}
Let us first prove that we can assume $\hat{\rho}_{ABC}$ pure.
Let $\hat{\rho}_{ABCR}$ be a purification of $\hat{\rho}_{ABC}$.
We have from the data-processing inequality for the quantum conditional entropy
\begin{equation}
S(Z|C)_{\hat{\rho}_{ZC}} \ge S(Z|CR)_{\hat{\rho}_{ZCR}}\;.
\end{equation}
Defining $C'=CR$, we have
\begin{equation}
S(Z|B)_{\hat{\rho}_{ZB}} + S(Z|C)_{\hat{\rho}_{ZC}} \ge S(Z|B)_{\hat{\rho}_{ZB}} + S(Z|C')_{\hat{\rho}_{ZC'}}\;,
\end{equation}
and we can then assume $\hat{\rho}_{ABC}$ pure.

For any $\mathbf{z}\in\mathbb{C}^M$, the state
\begin{equation}
\hat{\rho}_{BC|Z=\mathbf{z}} = \frac{\langle\mathbf{z}|\hat{\rho}_{ABC}|\mathbf{z}\rangle}{\langle\mathbf{z}|\hat{\rho}_{A}|\mathbf{z}\rangle}
\end{equation}
is also pure, hence
\begin{equation}
S(\hat{\rho}_{B|Z=\mathbf{z}}) = S(\hat{\rho}_{C|Z=\mathbf{z}})\;,
\end{equation}
and
\begin{equation}
S(B|Z)_{\hat{\rho}_{ZB}} = S(C|Z)_{\hat{\rho}_{ZC}}\;.
\end{equation}
We then have
\begin{align}
S(Z|C)_{\hat{\rho}_{ZC}} &= S(C|Z)_{\hat{\rho}_{ZC}} + S(\hat{\rho}_Z) - S(\hat{\rho}_C) = S(B|Z)_{\hat{\rho}_{ZB}} + S(\hat{\rho}_Z) - S(\hat{\rho}_{AB})\nonumber\\
&= S(Z|B)_{\hat{\rho}_{ZB}} - S(A|B)_{\hat{\rho}_{AB}}\;.
\end{align}
Finally, \autoref{thm:ineq} implies
\begin{align}\label{eq:proofbi}
S(Z|B)_{\hat{\rho}_{ZB}} + S(Z|C)_{\hat{\rho}_{ZC}} &= 2S(Z|B)_{\hat{\rho}_{ZB}} - S(A|B)_{\hat{\rho}_{AB}}\nonumber\\
&\ge M\ln\left(2+2\cosh\frac{S(A|B)_{\hat{\rho}_{AB}}}{M}\right) \ge M\ln4\;.
\end{align}
\end{proof}

\section{Optimality of the uncertainty relations}\label{sec:opt}
\begin{theorem}[optimality of the bipartite memory uncertainty relation]\label{thm:optimbi}
The uncertainty relation with bipartite memory \eqref{eq:biEUR} is optimal and the minimum \eqref{eq:biEUR} for $S(Z|B)$ is asymptotically achieved by a suitable sequence of quantum Gaussian states.
Indeed, let $A$ and $B$ be $M$-mode Gaussian quantum systems, and for any $a\ge1$ let $\hat{\gamma}_{AB}^{(a)}$ be the tensor product of $M$ two-mode squeezed pure quantum Gaussian states, with covariance matrix
\begin{equation}
\sigma_{AB}^{(a)} = \bigoplus_{i=1}^M \frac{1}{2}\left(
                \begin{array}{cc}
                  a\,I_2 & \sqrt{a^2-1}\,\sigma_Z \\
                  \sqrt{a^2-1}\,\sigma_Z & a\,I_2\\
                \end{array}
              \right)\;,
\end{equation}
where
\begin{equation}
I_2 = \left(
        \begin{array}{cc}
          1 & 0 \\
          0 & 1 \\
        \end{array}
      \right)\;,\qquad \sigma_Z = \left(
                         \begin{array}{cc}
                           1 & 0 \\
                           0 & -1 \\
                         \end{array}
                       \right)\;.
\end{equation}
Let us fix $s\in\mathbb{R}$, and let us define
\begin{equation}
\kappa = \exp\frac{s}{M} + 1\;.
\end{equation}
Then, the quantum Gaussian state
\begin{equation}\label{eq:defgamma}
\hat{\gamma}_{AB}^{(s,a)} = \left(\mathcal{A}_\kappa\otimes\mathbb{I}_B\right)\left(\hat{\gamma}_{AB}^{(a)}\right)
\end{equation}
satisfies
\begin{equation}
\lim_{a\to\infty} S(A|B)_{\hat{\gamma}_{AB}^{(s,a)}} = s\;,\qquad \lim_{a\to\infty}S(Z|B)_{\hat{\gamma}_{ZB}^{(s,a)}} = M\ln\left(\exp\frac{s}{M}+1\right)\;,
\end{equation}
where $\hat{\gamma}_{ZB}^{(s,a)}$ is the probability measure on $\mathbb{C}^M$ with values in positive operators on $B$ associated to the heterodyne measurement on $A$, defined as in \eqref{eq:defrhoZ}.
\end{theorem}
\begin{proof}
The quantum Gaussian state $\hat{\gamma}_{AB}^{(s,a)}$ has covariance matrix
\begin{equation}
\sigma_{AB}^{(s,a)} = \bigoplus_{i=1}^M\frac{1}{2}\left(
                \begin{array}{cc}
                  \left(\kappa\,a+\kappa-1\right)I_2 & \sqrt{\kappa\left(a^2-1\right)}\,\sigma_Z\\
                  \sqrt{\kappa\left(a^2-1\right)}\,\sigma_Z & a\,I_2\\
                \end{array}
              \right),
\end{equation}
whose symplectic eigenvalues are
\begin{equation}
\nu_+ = \frac{\kappa\,a + \kappa - a}{2}\;,\qquad \nu_- = \frac{1}{2}\;,
\end{equation}
each with multiplicity $M$.
We then have
\begin{equation}
S(A|B)_{\hat{\gamma}_{AB}^{(s,a)}} = M\,g\left(\frac{\kappa\,a + \kappa - a - 1}{2}\right) - M\,g\left(\frac{a - 1}{2}\right)\;,
\end{equation}
and
\begin{equation}
\lim_{a\to\infty}S(A|B)_{\hat{\gamma}_{AB}^{(s,a)}} = M\ln\left(\kappa-1\right) = s\;.
\end{equation}
From \cite{de2017wehrl}, Lemma 4, we have for any $\mathbf{z}\in\mathbb{C}^M$
\begin{equation}
\langle\mathbf{z}|\hat{\gamma}_{AB}^{(s,a)}|\mathbf{z}\rangle = \langle\mathbf{z}|\left(\mathcal{A}_\kappa\otimes\mathbb{I}_B\right)\left(\hat{\gamma}_{AB}^{(a)}\right)|\mathbf{z}\rangle = \frac{\langle\mathbf{z}/\sqrt{\kappa}|\hat{\gamma}_{AB}^{(a)}|\mathbf{z}/\sqrt{\kappa}\rangle}{\kappa^M}\;,
\end{equation}
where $|\mathbf{z}\rangle$ is the coherent state on $A$.
Then, since $\hat{\gamma}_{AB}^{(a)}$ is pure, also $\langle\mathbf{z}|\hat{\gamma}_{AB}^{(s,a)}|\mathbf{z}\rangle$ is pure and
\begin{equation}
S(B|Z)_{\hat{\gamma}_{ZB}^{(s,a)}} = 0\;.
\end{equation}
We have from \cite{de2017wehrl}, Eq. (70)
\begin{equation}
S\left(\hat{\gamma}_Z^{(s,a)}\right) = M\ln\frac{\kappa\left(a+1\right)}{2} + M\;,
\end{equation}
hence
\begin{align}
S(Z|B)_{\hat{\gamma}_{ZB}^{(s,a)}} &= S(B|Z)_{\hat{\gamma}_{ZB}^{(s,a)}} + S\left(\hat{\gamma}_Z^{(s,a)}\right) - S\left(\hat{\gamma}_B^{(s,a)}\right)\nonumber\\
&= M\ln\frac{\kappa\left(a+1\right)}{2} + M - M\,g\left(\frac{a-1}{2}\right)\;.
\end{align}
Finally,
\begin{equation}
\lim_{a\to\infty}S(Z|B)_{\hat{\gamma}_{ZB}^{(s,a)}} = M\ln\kappa = M\ln\left(\exp\frac{s}{M} + 1\right)\;.
\end{equation}
\end{proof}
\begin{theorem}[optimality of the tripartite memory uncertainty relation]\label{thm:optimtri}
The uncertainty relation with tripartite memory \eqref{eq:triEUR} is optimal and the value $M\ln4$ for $S(Z|A)+S(Z|B)$ is asymptotically achieved by a suitable sequence of quantum Gaussian states.
Indeed, for any $a\ge1$ let $\hat{\gamma}_{ABC}^{(0,a)}$ be the purification of the quantum Gaussian state $\hat{\gamma}_{AB}^{(0,a)}$ defined in \eqref{eq:defgamma}.
We then have
\begin{equation}
\lim_{a\to\infty}\left(S(Z|B)_{\hat{\gamma}_{ZB}^{(0,a)}} + S(Z|C)_{\hat{\gamma}_{ZC}^{(0,a)}}\right) = M\ln 4\;.
\end{equation}
\end{theorem}
\begin{proof}
From \autoref{thm:optimbi} we have
\begin{equation}
\lim_{a\to\infty}S(A|B)_{\hat{\gamma}_{AB}^{(0,a)}} = 0\;,\qquad \lim_{a\to\infty} S(Z|B)_{\hat{\gamma}_{ZB}^{(0,a)}} = M\ln 2\;.
\end{equation}
We then have from \eqref{eq:proofbi}
\begin{align}
\lim_{a\to\infty}\left(S(Z|B)_{\hat{\gamma}_{ZB}^{(0,a)}} + S(Z|C)_{\hat{\gamma}_{ZC}^{(0,a)}}\right) &= \lim_{a\to\infty}\left(2S(Z|B)_{\hat{\gamma}_{ZB}^{(0,a)}} - S(A|B)_{\hat{\gamma}_{AB}^{(0,a)}}\right) \nonumber\\
&= M\ln4\;.
\end{align}
\end{proof}

\section*{Acknowledgements}
I acknowledge financial support from the European Research Council (ERC Grant Agreements no 337603 and 321029), the Danish Council for Independent Research (Sapere Aude) and VILLUM FONDEN via the QMATH Centre of Excellence (Grant No. 10059).

\includegraphics[width=0.05\textwidth]{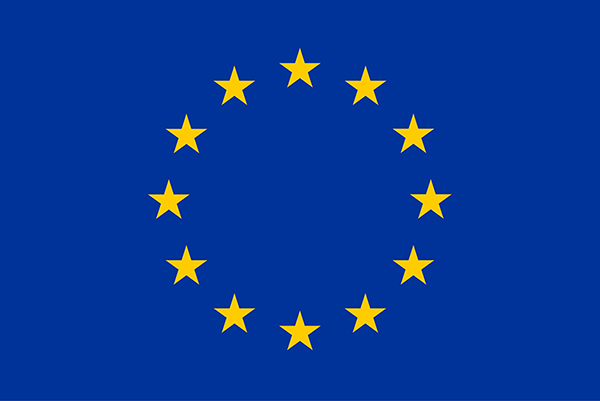}
This project has received funding from the European Union's Horizon 2020 research and innovation programme under the Marie Sk\l odowska-Curie grant agreement No. 792557.

\appendix
\section{}
\begin{lemma}[Jensen's trace inequality]\label{lem:jensen}
Let us consider the operator
\begin{equation}
\hat{A} = \sum_{k=0}^\infty a_k\,|\phi_k\rangle\langle\phi_k|\;,\qquad 0\le a_k\le 1\qquad\forall\;k\in\mathbb{N}\;,\qquad \sum_{k=0}^\infty a_k<\infty\;,
\end{equation}
where the vectors $|\phi_k\rangle$ form a resolution of the identity:
\begin{equation}\label{eq:eqphi}
\sum_{k=0}^\infty |\phi_k\rangle\langle\phi_k| = \hat{\mathbb{I}}\;.
\end{equation}
Then, for any concave function $f:[0,1]\to[0,\infty)$,
\begin{equation}
\mathrm{Tr}\,f\left(\hat{A}\right) \ge \sum_{k=0}^\infty f(a_k)\,\langle\phi_k|\phi_k\rangle\;.
\end{equation}
\end{lemma}
\begin{proof}
From \eqref{eq:eqphi} we get
\begin{equation}
\langle\phi_k|\phi_k\rangle \le 1\qquad\forall\;k\in\mathbb{N}\;.
\end{equation}
We then have
\begin{equation}
\mathrm{Tr}\,\hat{A} = \sum_{k=0}^\infty a_k\,\langle\phi_k|\phi_k\rangle \le \sum_{k=0}^\infty a_k < \infty\;.
\end{equation}
$\hat{A}$ has then discrete spectrum, and we can diagonalize it:
\begin{equation}
\hat{A} = \sum_{l=0}^\infty \lambda_l\,|v_l\rangle\langle v_l|\;,\qquad\langle v_k|v_l\rangle = \delta_{kl}\quad\forall\;k,\,l\in\mathbb{N}\;,\qquad \sum_{l=0}^\infty |v_l\rangle\langle v_l| = \hat{\mathbb{I}}\;.
\end{equation}
From the completeness of the $|\phi_k\rangle$, for any $l\in\mathbb{N}$
\begin{equation}
\sum_{k=0}^\infty |\langle v_l|\phi_k\rangle|^2 = \langle v_l|v_l\rangle =1\;,
\end{equation}
hence $|\langle v_l|\phi_k\rangle|^2$ is a probability distribution on $\mathbb{N}$.
We then have from Jensen's inequality
\begin{align}
\mathrm{Tr}\,f\left(\hat{A}\right) &= \sum_{l=0}^\infty f(\lambda_l) = \sum_{l=0}^\infty f\left(\sum_{k=0}^\infty |\langle v_l|\phi_k\rangle|^2 a_k\right) \ge \sum_{k,\,l=0}^\infty |\langle v_l|\phi_k\rangle|^2\,f(a_k)\nonumber\\
&= \sum_{k=0}^\infty \langle \phi_k|\phi_k\rangle\,f(a_k)\;,
\end{align}
where we have used that for the completeness of the $|v_l\rangle$, for any $k\in\mathbb{N}$
\begin{equation}
\sum_{l=0}^\infty |\langle v_l|\phi_k\rangle|^2 = \langle\phi_k|\phi_k\rangle\;.
\end{equation}
\end{proof}

\bibliographystyle{spmpsci}
\bibliography{biblio}

\end{document}